\documentclass[submit]{scias}
\usepackage{pxfonts}
\newdimen\proofrulebreadth \proofrulebreadth=.05em
\newdimen\proofdotseparation \proofdotseparation=1.25ex
\newdimen\proofrulebaseline \proofrulebaseline=2ex
\newcount\proofdotnumber \proofdotnumber=3
\let\then\relax
\def\hfi{\hskip0pt plus.0001fil}
\mathchardef\squigto="3A3B
\newif\ifinsideprooftree\insideprooftreefalse
\newif\ifonleftofproofrule\onleftofproofrulefalse
\newif\ifproofdots\proofdotsfalse
\newif\ifdoubleproof\doubleprooffalse
\let\wereinproofbit\relax
\newdimen\shortenproofleft
\newdimen\shortenproofright
\newdimen\proofbelowshift
\newbox\proofabove
\newbox\proofbelow
\newbox\proofrulename
\def\shiftproofbelow{\let\next\relax\afterassignment\setshiftproofbelow\dimen0 }
\def\shiftproofbelowneg{\def\next{\multiply\dimen0 by-1 }%
\afterassignment\setshiftproofbelow\dimen0 }
\def\setshiftproofbelow{\next\proofbelowshift=\dimen0 }
\def\setproofrulebreadth{\proofrulebreadth}

\def\prooftree{
\ifnum  \lastpenalty=1
\then   \unpenalty
\else   \onleftofproofrulefalse
\fi
\ifonleftofproofrule
\else   \ifinsideprooftree
        \then   \hskip.5em plus1fil
        \fi
\fi
\bgroup
\setbox\proofbelow=\hbox{}\setbox\proofrulename=\hbox{}%
\let\justifies\proofover\let\leadsto\proofoverdots\let\Justifies\proofoverdbl
\let\using\proofusing\let\[\prooftree
\ifinsideprooftree\let\]\endprooftree\fi
\proofdotsfalse\doubleprooffalse
\let\thickness\setproofrulebreadth
\let\shiftright\shiftproofbelow \let\shift\shiftproofbelow
\let\shiftleft\shiftproofbelowneg
\let\ifwasinsideprooftree\ifinsideprooftree
\insideprooftreetrue
\setbox\proofabove=\hbox\bgroup$\displaystyle 
\let\wereinproofbit\prooftree
\shortenproofleft=0pt \shortenproofright=0pt \proofbelowshift=0pt
\onleftofproofruletrue\penalty1
}

\def\eproofbit{
\ifx    \wereinproofbit\prooftree
\then   \ifcase \lastpenalty
        \then   \shortenproofright=0pt  
        \or     \unpenalty\hfil         
        \or     \unpenalty\unskip       
        \else   \shortenproofright=0pt  
        \fi
\fi
\global\dimen0=\shortenproofleft
\global\dimen1=\shortenproofright
\global\dimen2=\proofrulebreadth
\global\dimen3=\proofbelowshift
\global\dimen4=\proofdotseparation
\global\count255=\proofdotnumber
$\egroup  
\shortenproofleft=\dimen0
\shortenproofright=\dimen1
\proofrulebreadth=\dimen2
\proofbelowshift=\dimen3
\proofdotseparation=\dimen4
\proofdotnumber=\count255
}

\def\proofover{
\eproofbit 
\setbox\proofbelow=\hbox\bgroup 
\let\wereinproofbit\proofover
$\displaystyle
}%
\def\proofoverdbl{
\eproofbit 
\doubleprooftrue
\setbox\proofbelow=\hbox\bgroup 
\let\wereinproofbit\proofoverdbl
$\displaystyle
}%
\def\proofoverdots{
\eproofbit 
\proofdotstrue
\setbox\proofbelow=\hbox\bgroup 
\let\wereinproofbit\proofoverdots
$\displaystyle
}%
\def\proofusing{
\eproofbit 
\setbox\proofrulename=\hbox\bgroup 
\let\wereinproofbit\proofusing
\kern0.3em$
}

\def\endprooftree{
\eproofbit 
  \dimen5 =0pt
\dimen0=\wd\proofabove \advance\dimen0-\shortenproofleft
\advance\dimen0-\shortenproofright
\dimen1=.5\dimen0 \advance\dimen1-.5\wd\proofbelow
\dimen4=\dimen1
\advance\dimen1\proofbelowshift \advance\dimen4-\proofbelowshift
\ifdim  \dimen1<0pt
\then   \advance\shortenproofleft\dimen1
        \advance\dimen0-\dimen1
        \dimen1=0pt
        \ifdim  \shortenproofleft<0pt
        \then   \setbox\proofabove=\hbox{%
                        \kern-\shortenproofleft\unhbox\proofabove}%
                \shortenproofleft=0pt
        \fi
\fi
\ifdim  \dimen4<0pt
\then   \advance\shortenproofright\dimen4
        \advance\dimen0-\dimen4
        \dimen4=0pt
\fi
\ifdim  \shortenproofright<\wd\proofrulename
\then   \shortenproofright=\wd\proofrulename
\fi
\dimen2=\shortenproofleft \advance\dimen2 by\dimen1
\dimen3=\shortenproofright\advance\dimen3 by\dimen4
\ifproofdots
\then
        \dimen6=\shortenproofleft \advance\dimen6 .5\dimen0
        \setbox1=\vbox to\proofdotseparation{\vss\hbox{$\cdot$}\vss}%
        \setbox0=\hbox{%
                \advance\dimen6-.5\wd1
                \kern\dimen6
                $\vcenter to\proofdotnumber\proofdotseparation
                        {\leaders\box1\vfill}$%
                \unhbox\proofrulename}%
\else   \dimen6=\fontdimen22\the\textfont2 
        \dimen7=\dimen6
        \advance\dimen6by.5\proofrulebreadth
        \advance\dimen7by-.5\proofrulebreadth
        \setbox0=\hbox{%
                \kern\shortenproofleft
                \ifdoubleproof
                \then   \hbox to\dimen0{%
                        $\mathsurround0pt\mathord=\mkern-6mu%
                        \cleaders\hbox{$\mkern-2mu=\mkern-2mu$}\hfill
                        \mkern-6mu\mathord=$}%
                \else   \vrule height\dimen6 depth-\dimen7 width\dimen0
                \fi
                \unhbox\proofrulename}%
        \ht0=\dimen6 \dp0=-\dimen7
\fi
\let\doll\relax
\ifwasinsideprooftree
\then   \let\VBOX\vbox
\else   \ifmmode\else$\let\doll=$\fi
        \let\VBOX\vcenter
\fi
\VBOX   {\baselineskip\proofrulebaseline \lineskip.2ex
        \expandafter\lineskiplimit\ifproofdots0ex\else-0.6ex\fi
        \hbox   spread\dimen5   {\hfi\unhbox\proofabove\hfi}%
        \hbox{\box0}%
        \hbox   {\kern\dimen2 \box\proofbelow}}\doll%
\global\dimen2=\dimen2
\global\dimen3=\dimen3
\egroup 
\ifonleftofproofrule
\then   \shortenproofleft=\dimen2
\fi
\shortenproofright=\dimen3
\onleftofproofrulefalse
\ifinsideprooftree
\then   \hskip.5em plus 1fil \penalty2
\fi
}

\def\df{{\stackrel{\mathrm{def}}{\Longleftrightarrow}}}

\newcommand{\pts}{\mathit{pts}}
\newcommand{\ptsp}{\mathit{pts}^\prime}
\newcommand{\lb}{\llbracket}
\newcommand{\rb}{\rrbracket}
\newcommand{\ra}{\rightarrow}
\newcommand{\Ra}{\Longrightarrow}
\newcommand{\while}{\mathit{while\ b\ do\ } S_t}
\newcommand{\ifs}{\mathit{if\ b\ then\ }S_t \mathit{\ else\ }S_f}

\title{Probabilistic pointer analysis for multithreaded programs}
\author{Mohamed A. El-Zawawy} 
\address{College of Computer and Information Sciences\\
Al-Imam M. I.-S. I. University\\ Riyadh 11432\\Kingdom of Saudi
Arabia
\\ and \\Department of
Mathematics\\ Faculty of Science \\Cairo University\\ Giza 12613\\
Egypt}

\ead{maelzawawy@cu.edu.eg} 

\abstract{ The use of pointers and data-structures based on pointers
results in circular memory references that are interpreted by a
vital compiler analysis, namely pointer analysis. For a pair of
memory references at a program point, a typical pointer analysis
specifies if the points-to relation between them may exist,
definitely does not exist, or definitely exists. The "may be" case,
which describes the points-to relation for most of the pairs, can
not be dealt with by most compiler optimizations. This is so to
guarantee the soundness of these optimizations. However the "may be"
case can be capitalized by the modern class of speculative
optimizations if the probability that two memory references alias
can be measured. Focusing on multithreading, a prevailing technique
of programming, this paper presents a new flow-sensitive technique
for probabilistic pointer analysis of multithreaded programs. The
proposed technique has the form of a type system and calculates the
probability of every points-to relation at each program point. The
key to our approach is to calculate the points-to information via a
post-type derivation. The use of type systems has the advantage of
associating each analysis results with a justification (proof) for
the correctness of the results. This justification has the form of a
type derivation and is very much required in applications like
certified code. }

\keywords{Static analysis, Speculative optimizations, Probabilistic
alias analysis, Distributed programs, Semantics of multithreaded
programs, Type systems.}

\begin{document}
\frontmatter

\section{Introduction}\label{intro}

Multithreading is enjoying a growing interest and becoming a
prevailing technique of programming. The use of multiple threads has
several advantages: (a) concealing the delay of commands like
reading from a secondary storage (b) improving the action of
programs, like web servers, that run on multiprocessors, (c)
building complex systems for user interface, (d) simplifying the
process of organizing huge systems of code. However the static
analysis of multithreaded programs~\cite{Gelado10,Leung09,Xiao10} is
intricate due to the possible interaction between multiple threads.

Among effective tools of modern programming languages are pointers
which empower coding intricate data structures. Not only does the
uncertainty of pointer values at compile time complicate analysis of
programs, but also retard program compilation by compelling the
program optimization and analysis to be conservative. The pointer
analysis~\cite{El-Zawawy11,El-Zawawy11-2,Adams02} of programs is a
challenging problem in which researchers have trade space and time
costs for precision. However binary decision
diagrams~\cite{Berndl03} have been used to ease the difficulty of
this trade off.

At any program point and for every pair of memory references, a
traditional pointer analysis figures out whether one of these
references may point to, definitely points to, or definitely does
point to the other reference. For most of pairs of the memory
references the points-to relation is of type "may be". This is
specially the case for techniques that prefer speed over accuracy.
Traditional optimization techniques are not robust enough to treat
the cases "may be" and "definitely" differently. The idea behind
speculative optimization is to subsidize the "maybe" case, specially
if the probability of "maybe" cab be specifically
quantified~\cite{Chen04,Silva06}.

Pointer analysis~\cite{Yu10,Anderson02} is among most important
program analyses of multithreaded programs. Pointer analysis of
multithreaded programs has many applications; (a) mechanical binding
of file operations that are in abeyance, (b) optimizations for
memory systems like prefetching and relocating remote data
calculations, (c) equipping compilers with necessary information for
optimizations like common subexpression elimination and induction
variable elimination, and (d) relaxing the process of developing
complex tools for software engineering like program slicers and race
detectors.

This paper presents a new technique for pointer analysis of
multithreaded programs. The proposed technique is probabilistic; it
anticipates precisely for every program point the probability of
every points-to relation. Building on a type system, the proposed
approach is control-flow-sensitive. The key to the presented
analysis is to calculate probabilities for points-to relations
through the compositional use of inference rules of a type system.
The proposed technique associates with every analysis a proof (type
derivation) for the correctness of the analysis.

Among techniques to approach static analysis of programs is the
algorithmic style. However the proposed technique of this paper has
the form of a type system. The algorithmic style does not reflect
how the analysis results are obtained because it works on
control-flow graphs of programs; not on phrase structures as in the
case of type systems. Therefore the type-systems
approach~\cite{Bertot04,Laud06,Saabas08,El-Zawawy11} is perfect for
applications that require to handle a justifications (proof) for
correctness of analsys results together with each individual
analysis. An example of such applications is certified code. What
contributes to suitability of type-systems tools to produce such
proofs is the relative simplicity of its inference rules. This
simplicity is a much appreciated property in applications that
require justifications. In type-systems approach, the justifications
take the form of type derivations.

\subsection*{Motivation}
\begin{figure}[thp]
{\footnotesize{\begin{center}
\begin{tabular}{l}
  $1.\quad a\coloneqq \& c;$ \\
  $2.\quad if (\ldots) \textit{ then } b\coloneqq
  \&c $\\
  $3.\hspace{1.2cm}\textit{ else } b\coloneqq  \&d;$ \\
  $4.\quad \textit{par}\{$ \\
  $5.\hspace{.8cm}  \{a\coloneqq \& c\}$  \\
  $6.\hspace{.8cm}  \{a\coloneqq \& d\}$\\
  $7.\hspace{.8cm}\};$  \\
  $8.\quad\textit{while}(\ldots)$  \\
  $9.\hspace{.8cm} if (\ldots) \textit{ then } e\coloneqq
  \&d $\\
  $10.\hspace{1.5cm}\textit{ else } e\coloneqq  5;$ \\
   \end{tabular}\end{center}
\caption{A motivating example.}\label{example1} }}\end{figure}
\begin{figure*}
\begin{center}\begin{tabular}{|l|l|} \hline
Program point   & Pointer information \\
\hline
first point& $\{t\mapsto\emptyset\mid t\in\textit{Var}\}$ \\
between lines 1 \& 2   &
$\{a\mapsto\{(c^\prime,1)\},t\mapsto\emptyset\mid x\not= t\}$   \\
point between 3 \& 4   &
$\{a\mapsto\{(c^\prime,1)\},b\mapsto\{(c^\prime,0.6),(d^\prime,0.4)\},$\\
& $t\mapsto\emptyset\mid t\notin\{a,b\}\}$ \\
point between 7 \& 8
&$\{a\mapsto\{(c^\prime,0.5),(d^\prime,0.5)\},$\\
& $b\mapsto\{(c^\prime,0.6),(d^\prime,0.4)\},t\mapsto\emptyset
\mid t\notin\{a,b\}\}$ \\
last point &$\{a\mapsto\{(c^\prime,0.5),(d^\prime,0.5)\},
e\mapsto\{(d^\prime,\frac{1}{100}\times
\Sigma_{i=1}^{i=100}(\frac{1}{2})^i)\}$\\
& $b\mapsto\{(c^\prime,0.6),(d^\prime,0.4)\},t\mapsto\emptyset
\mid t\notin\{a,b,c\}\}$ \\
\hline
\end{tabular}\end{center}
    \caption{Results of pointer analysis of program in Figure 1.}\label{pointer result}
\end{figure*}
Figure~\ref{example1} presents a motivating example of our work.
This example uses three pointer variables ($a$, $b$, and $e$) that
point at two variables ($c$ and $d$). We suppose that (i) the
condition of the \textit{if} statement at line $2$ is true with
probability $0.6$, (ii) the condition of the \textit{if} statement
at line $9$ is true with probability $0.5$, and (iii) the loop at
line $8$ iterates at most $100$ times. These statistical and
probabilistic information can be obtained using edge
profiling~\cite{Bond05,Anderson97,Suganuma05,Vaswani05}. In absence
of edge profiling, heuristics can be used. The work presented in
this paper aims at introducing a probabilistic pointer analysis that
produces results like that in Figure~\ref{pointer result}. The aim
is also to associate each such pointer-analysis result with a
justification for the correctness of the result. This justification
takes the form of a type derivation in our proposed technique which
is based on a type system.

\subsection*{Contributions}
Contributions of this paper are the following:
\begin{enumerate}
\item A new pointer analysis technique, that is probabilistic and
flow-sensitive, for multithreaded programs.
\item A new probabilistic operational-semantics for multithreaded programs.
\end{enumerate}

\subsubsection*{Organization}
The remainder of the paper is organized in three sections as
follows. The first of these sections presents a simple language
equipped with parallel and pointer constructs. This section also
presents a new probabilistic operational semantics for the
constructs of the language that we study. The second of these
sections introduces a type system to carry probabilistic pointer
analysis of parallel programs. This involves introducing suitable
notions for pointer types, a subtyping relation, and a detailed
proof for the soundness of the proposed type system w.r.t. the
semantics presented in the paper. Related work is reviewed in the
last section of the paper.

\section{Probabilistic operational semantics}\label{semantics}
This section presents the programming language we study and a
probabilistic pointer analysis for its constructs. We build our
language (Figure~\ref{lang}) on the \textit{while} language,
originally presented by Hoare in 1969, by equipping it with commands
dealing with pointers and parallel computations. The parallel
concepts dealt with in our language are fork-join, conditionally
spawned threads, and parallel loops. These concepts are represented
by commands \textit{par}, \textit{par-if}, and \textit{par-for}),
respectively. Sates of our proposed operational semantics are
defined as follows:
\begin{definition}
\begin{enumerate}
\item $\textit{Addrs}=\{x^\prime\mid x\in \textit{Var}\}$ and $
{{\textit{Val}} = {\mathbb{Z} \cup \textit{Addrs}}}$.
\item $\gamma\in\Gamma= \textit{Var}\rightarrow \textit{Val}$.
\item $\textit{state} \in\textit{States}={\{(\gamma,p)
\mid \gamma\in \Gamma \wedge p\in[0,1]\}}\cup\{\textit{abort}\}$.
\end{enumerate}
\end{definition}
Typically, a state is a function from the set of variables to the
set of values (integers). In our work, we enrich the set of values
with a set of symbolic addresses and enrich each state with a
probabilistic value that is meant to measure the probability with
which this state is reached. The \textit{abort} state is there to
capture any case of de-reference that is unsafe; i.e de-referencing
a variable that contains no address. We assume that the set of
program variables, $\textit{Var}$, is finite.
\begin{figure*}[h]
\begin{eqnarray*}
&  & {n}\in {\mathbb{Z}},\ {x}\in {\textit{Var}},\ \textit{and}\
{\oplus}
\in{\{+,-,\times\}}\\
e \in \textit{Aexprs}  &  \Coloneqq & {x}\mid  {n} \mid {e_1 \oplus e_2}\\
b\in \textit{Bexprs} &  \Coloneqq & {\textit{true}} \mid
{\textit{false}} \mid {\neg b}\mid {e_1 = e_2} \mid {e_1 \leq e_2}
\mid{b_1 \wedge b_2}
\mid {b_1 \vee b_2} \\
S\in \textit{Stmts} &  \Coloneqq & {x\coloneqq e}\mid {x\coloneqq \&
y}\mid {*x\coloneqq e}\mid {x\coloneqq *
y}\mid{\textit{skip}}\mid {S_1;S_2}\mid {\ifs} \mid \\
&  & {\while}\mid \textit{par}\{\{S_1\},\ldots,\{S_n\}\}\mid
\textit{par-if}\{(b_1,S_1),\ldots,(b_n,S_n)\}\mid
\textit{par-for}\{S\}.
\end{eqnarray*}
\caption{The programming language.}\label{lang}
\end{figure*}

Except that arithmetic and Boolean operations are not allowed on
pointers, the semantics of arithmetic and Boolean expressions are
defined as usual (Figure~\ref{sem}). The inference rules of
Figure~\ref{inf} define the transition relation $\rightsquigarrow$
of our operational semantics.

\begin{figure*}
{\footnotesize{
\[
{\lb n\rb \gamma} = {n}\quad {\lb \& x\rb \gamma} = {x^\prime}\quad
{\lb x\rb \gamma} = {\gamma(x)}\quad {\lb \textit{true}\rb \gamma} =
{\textit{true}}\quad {\lb \textit{false}\rb \gamma} =
{\textit{false}}
\]\[
{\lb *x\rb \gamma} = \left\{
\begin{array}{ll}
\gamma(y) &\mbox{if }{\gamma(x)}={y^\prime}, \\
! & \mbox{otherwise.}
\end{array} \right.
\quad {\lb e_1\oplus e_2\rb \gamma} = \left\{\begin{array}{ll} {\lb
e_1\rb \gamma}\oplus {\lb e_2\rb \gamma} & \mbox{if }{\lb e_1\rb
\gamma,\lb e_2\rb \gamma}\in{\mathbb{Z}}, \\
! &\mbox{otherwise.}
\end{array}
\right.
\]\[
{\lb \neg A\rb \gamma} = \left\{
\begin{array}{ll}
{\neg (\lb A\rb \gamma)} &\mbox{if }{\lb A\rb \gamma} \in
{\{\textit{true}, \textit{false}\},} \\! &\mbox{otherwise.}
\end{array} \right. \quad
{\lb e_1= e_2\rb \gamma} =\left\{\begin{array}{ll} ! & \mbox{if }
{\lb e_1\rb \gamma}= {!}
\mbox{ or } {\lb e_2\rb \gamma}= {!}, \\
\textit{true} & \mbox{if }{\lb e_1\rb \gamma}= { \lb e_2\rb
\gamma}\not = {!},\\ \textit{false} & \mbox{otherwise}.
\end{array}\right.
\]\[
{\lb e_1\le e_2\rb \gamma} =\left\{\begin{array}{ll} ! & \mbox{if }
{\lb e_1\rb \gamma}\not \in \mathbb{Z}
\mbox{ or } {\lb e_2\rb \gamma}\not \in \mathbb{Z}, \\
{\lb e_1\rb \gamma}\le { \lb e_2\rb \gamma} & \mbox{otherwise}.
\end{array}\right.
\]\[
\mbox{For }{\diamond} \in{\{\wedge,\vee\}},\  {\lb b_1\diamond
b_2\rb \gamma} =\left\{\begin{array}{ll} ! & \mbox{if } {\lb b_1\rb
\gamma}= {!}
\mbox{ or } {\lb b_2\rb \gamma}= {!}, \\
{\lb b_1\rb \gamma}\diamond { \lb b_2\rb \gamma} & \mbox{otherwise}.
\end{array}\right.
\]}}
  \caption{Semantics of arithmetic and Boolean expressions.}\label{sem}
\end{figure*}
\begin{figure*}
{\footnotesize{
\[
\begin{prooftree}
{\lb e\rb\gamma} = {!}\justifies  x \coloneqq e:
(\gamma,p)\rightsquigarrow \textit{abort}\thickness=0.08em
\end{prooftree}\quad
\begin{prooftree}
{\lb e\rb\gamma} \not= {!}\justifies  x \coloneqq e:(\gamma,p)
\rightsquigarrow (\gamma[x\mapsto\lb e\rb\gamma],p)\thickness=0.08em
\end{prooftree} \quad
\begin{prooftree} \gamma(x) =z^\prime \quad z \coloneqq e:(\gamma,p)
\rightsquigarrow \textit{state}\justifies *x \coloneqq e:
(\gamma,p)\rightsquigarrow \textit{state}\thickness=0.08em
\end{prooftree}
\]\[
\begin{prooftree} \gamma(x)\notin \textit{Addrs}
\justifies *x \coloneqq e: (\gamma,p)\rightsquigarrow
\textit{abort}\thickness=0.08em \end{prooftree}\quad
\begin{prooftree}
\justifies  x \coloneqq \& y:(\gamma,p) \rightsquigarrow
(\gamma[x\mapsto y^\prime],p)\thickness=0.08em
\end{prooftree} \quad\begin{prooftree} \gamma(y)\notin \textit{Addrs}
 \justifies x \coloneqq *y: (\gamma,p)\rightsquigarrow
\textit{abort}\thickness=0.08em \end{prooftree}
\]\[
\begin{prooftree} \gamma(y) =z^\prime \quad x \coloneqq z:
(\gamma,p)\rightsquigarrow (\gamma^\prime,p) \justifies x \coloneqq
* y: (\gamma,p)\rightsquigarrow (\gamma^\prime,p) \thickness=0.08em
\end{prooftree} \quad
\begin{prooftree}
\justifies  \textit{skip}: (\gamma,p) \rightsquigarrow
(\gamma,p)\thickness=0.08em
\end{prooftree}
\quad
\begin{prooftree}
S_1:(\gamma,p) \rightsquigarrow \textit{abort}\justifies S_1;S_2:
(\gamma,p) \rightsquigarrow \textit{abort}\thickness=0.08em
\end{prooftree}
\]\[
\begin{prooftree}
S_1:(\gamma,p) \rightsquigarrow
(\gamma^{\prime\prime},p^{\prime\prime})\quad
S_2:(\gamma^{\prime\prime},p^{\prime\prime})\rightsquigarrow
\textit{state}\justifies S_1;S_2: (\gamma,p) \rightsquigarrow
\textit{state}\thickness=0.08em
\end{prooftree}
\quad
\begin{prooftree}
{\lb b\rb \gamma}= {!} \justifies \ifs:(\gamma,p)
\rightsquigarrow\textit{abort}\thickness=0.08em
\end{prooftree}
\]\[
\begin{prooftree}
{\lb b\rb \gamma}= {\textit{true}}  \quad S_t:(\gamma,p)
\rightsquigarrow \textit{abort} \justifies \ifs:(\gamma,p)
\rightsquigarrow \textit{abort}\thickness=0.08em
\end{prooftree}
\quad
\begin{prooftree}
{\lb b\rb \gamma}= {\textit{true}}  \quad S_t:(\gamma,p)
\rightsquigarrow (\gamma^\prime,p^\prime) \justifies \ifs:(\gamma,p)
\rightsquigarrow (\gamma^\prime,p_{if}\times
p^\prime)\thickness=0.08em
\end{prooftree}
\]\[
\begin{prooftree}
{\lb b\rb \gamma}= {\textit{false}}  \quad S_f:(\gamma,p)
\rightsquigarrow \textit{abort} \justifies \ifs:(\gamma,p)
\rightsquigarrow \textit{abort}\thickness=0.08em
\end{prooftree}
\quad
\begin{prooftree}
{\lb b\rb \gamma}= {\textit{false}}  \quad S_f:(\gamma,p)
\rightsquigarrow (\gamma^\prime,p^\prime) \justifies \ifs:(\gamma,p)
\rightsquigarrow (\gamma^\prime,(1-p_{if})\times
p^\prime)\thickness=0.08em
\end{prooftree}
\]\[
\begin{prooftree} {\lb b\rb \gamma}= {!} \justifies
\while:(\gamma,p) \rightsquigarrow \textit{abort}\thickness=0.08em
\end{prooftree}
\quad
\begin{prooftree} {\lb b\rb \gamma}= {\textit{false}}
\justifies \while:(\gamma,p)
\rightsquigarrow(\gamma,p)\thickness=0.08em
\end{prooftree}
\]\[
\begin{prooftree}
{\lb b\rb \gamma}= {\textit{true}}  \quad S:(\gamma,p)
\rightsquigarrow \textit{abort} \justifies \while:(\gamma,p)
\rightsquigarrow \textit{abort}\thickness=0.08em
\end{prooftree}
\]\[
\begin{prooftree}
{\lb b\rb \gamma}= {\textit{true}}  \quad S:(\gamma,p)
\rightsquigarrow (\gamma^{\prime\prime},p^{\prime\prime}) \quad
\while:(\gamma^{\prime\prime},p^{\prime\prime}) \rightsquigarrow
\textit{state} \justifies \while:(\gamma,p) \rightsquigarrow
\textit{state}\thickness=0.08em
\end{prooftree}
\]}}
\begin{itemize}
\item[$\bullet$] \textbf{Fork-join}: {\footnotesize{\[
\begin{prooftree}
(\exists\ \theta:\{1,\ldots,n\}\ra\{1,\ldots,n\}).\
S_{\theta(1)};S_{\theta(2)};\ldots;S_{\theta(n)}:(\gamma,p)
\rightsquigarrow(\gamma^\prime,p^\prime) \justifies
\textit{par}\{\{S_1\},\ldots,\{S_n\}\}:(\gamma,p)
\rightsquigarrow(\gamma^\prime,\frac{1}{n!}\times
p^\prime)\thickness=0.08em\using{(\textit{par-sem})}
\end{prooftree}\]\[
\begin{prooftree}
(\exists\ \theta:\{1,\ldots,n\}\ra\{1,\ldots,n\}).\
S_{\theta(1)};S_{\theta(2)};\ldots;S_{\theta(n)}:(\gamma,p)
\rightsquigarrow\textit{abort} \justifies
\textit{par}\{\{S_1\},\ldots,\{S_n\}\}:(\gamma,p)
\rightsquigarrow\textit{abort}\thickness=0.08em
\end{prooftree}\]
}}
\item[$\bullet$] \textbf{Conditionally spawned threads}:
{\footnotesize{
\[\begin{prooftree}
\textit{par}\{\{\mathit{if\ b_1\ then\ }S_1 \mathit{\ else\
skip}\},\ldots,\{\mathit{if\ b_n\ then\ }S_n \mathit{\ else\
skip}\}\}:(\gamma,p)\rightsquigarrow \textit{state}\justifies
\textit{par-if}\{(b_1,S_1),\ldots,(b_n,S_n)\}:(\gamma,p)
\rightsquigarrow \textit{state}\thickness=0.08em
\end{prooftree}\]}}
\item[$\bullet$] \textbf{Parallel loops}: {\footnotesize{\[
\begin{prooftree}
\exists n.\
\textit{par}\{\overbrace{\{S\},\ldots,\{S\}}^{n-times}\}:
(\gamma,p)\rightsquigarrow\textit{state}\justifies
\textit{par-for}\{S\}:(\gamma,p) \rightsquigarrow
\textit{state}\thickness=0.08em
\end{prooftree}\]
}} \end{itemize} \caption{Inference rules of the
semantics.}\label{inf}
\end{figure*}
We notice that none of the assignment statements changes the
probability component of a given pre-state to produce the
corresponding post-state. The symbol $p_{if}$ used in the inference
rules of the \textit{if} statement denotes a number in $[0,1]$ and
measures the probability that the condition of the statement is
true. This probabilistic information can be obtained using edge
profiling~\cite{Bond05,Anderson97,Suganuma05,Vaswani05}. In absence
of edge profiling, heuristics can be used.

The \textit{par} command is the main parallel concept. This concept
is also known as cobegin-coend or fork-join. The execution of this
command amounts to starting concurrently executing the threads of
the command at the beginning of the construct and then to wait for
the completion of these executions at the end of the construct. Then
the subsequent command can be executed. The inference rule
(\textit{par-sem}) approximates the execution methodology of the
\textit{par} command. The probability $p^\prime$ in the rule
(\textit{par-sem}) is multiplied by $\frac{1}{n!}$ (not by
$\frac{1}{n}$ as the reader may expect) because the permutation
$\theta$ finds one of the $n!$ ways in which the threads can be
sorted and then executed. As an example, the reader may consider
applying the rule \textit{par-sem} when $n=3$ and the threads are
$S_1: a\coloneqq b+c, S_2: b\coloneqq a\times c, \hbox{and } S_3:
c\coloneqq a-b$. The semantics of \textit{par-if} and
\textit{par-for} commands are defined using that of the \textit{par}
command.

\section{Probabilistic pointer analysis}\label{pointers}
The purpose of a typical pointer analysis is to assign to every
program point a points-to function. The domain of this function is
the set of all pairs of pointers and the codomain is the set
$\{\textit{definitely exists, definitely does not exist, may
exist}\}.$ The codomain describes the points-to relation between
pairs of memory references. For most of the pointer pairs, the
points-to relation is "may exist". This is specially the case for
techniques of pointer analysis that give priority for speed over
efficiency. The common drawback for most existing program
optimization techniques is that  they can not treat the "maybe" and
"definitely does not exist" cases differently. Speculative
optimizations are meant to overcome this disadvantage via working on
the result of analyses that can measure the probability that a
points-to  relation exist between two pointers.

This section presents a new technique for probabilistic pointer
analysis for multithreaded programs. The technique has the form of a
type system and its goal is to accurately calculate the likelihood
at each program point for every points-to relation. The advantages
of the proposed technique include the simplicity of the inference
rules of the type system and that no dependence profile information
(information describing dependencies between threads) is required.
Dependence profile information, required by some multithreading
techniques like~\cite{Steffan05}, is expensive to get. The proposed
technique is flow-sensitive. The key to our technique is to
calculate points-to probabilities via a post type derivation for a
given program using the bottom points-to type as a pre type.

The following definition presents some notations that are used in
the rest of the paper.

\begin{definition}\label{def2}
\begin{enumerate}
\item $\textit{Addrs} = \{ x^\prime\mid x\in \textit{Var}\}$ and
${\textit{Addrs}_p} = {\textit{Addrs}\times [0,1]}$.
\item $\textit{Pre-PTS} = \{ \pts\mid \pts:\textit{Var}
\ra 2^{\textit{Addrs}_p} $ s.t. $ {\forall y\in \textit{Var}.}\
(y^\prime,p_1),(y^\prime,p_2)\in \pts(x)\Ra p_1=p_2\}$.
\item For $\pts\in\textit{Pre-PTS}$ and $x\in \textit{Var},$ $
{{\sum_\pts x}={\sum_{(z^\prime,p)\in\pts(x)}p}}$.
\item For every $\pts\in\textit{Pre-PTS}$ and $x\in \textit{Var},$ $
{A_\pts(x)=\{z^\prime\mid \exists p>0.\ (z^\prime,p)\in\pts(x)\}}$.
\item For $A\in\textit{Addrs}_p,\pts\in \textit{Pre-PTS},$ and $0\le q\le 1,$
    \begin{enumerate}
    \item ${A\times q}={\{(y^\prime,p\times q)\mid (y^\prime,p)\in A\}}$.
    \item $\pts\times q$ is the function defined by
    ${(\pts\times q)(x)} = {\pts(x)\times q}.$
    \end{enumerate}
\end{enumerate}
\end{definition}

We note that the set of symbolic addresses \textit{Addrs} is
enriched with probabilities to form the set $\textit{Addrs}_p$. In
line with real situations, the condition on the elements of
\textit{Pre-PTS} excludes maps that assign the same address for a
variable with two different probabilities. The notation $\sum_\pts
x$ denotes the probability that the variables $x$ has an address
with respect to $\pts$. The notation $A_\pts(x)$ denotes the set of
addresses that have a non-zero probability to get into $x$. The
multiplication operations of Definition~\ref{def2}.5 are necessary
to join many points-to types (each with a different probability)
into one type.

A formalization for the concepts of the set of points-to types
\textit{PTS}, the subtyping relation $\le$, and the relation
${\models}\subseteq {\Gamma\times \textit{PTS}}$ are in the
subsequent definition.

\begin{definition}\label{order}
\begin{enumerate}
\item $\mathit{PTS} = \{ \pts\in \textit{Pre-PTS}\mid {\forall x\in\textit{Var}.\
\sum_\pts x \leq 1}\}$.
\item $\pts \le \ptsp\ \df\ \forall x.\ A_\pts(x)\subseteq A_{\pts^\prime}(x)$.
\item $\pts \equiv \ptsp\ \df\ \forall x.\ A_\pts(x)= A_{\pts^\prime}(x)$.
\item $(\gamma,p)\models\pts \ \df\ (\forall x.\
\gamma(x)\in \textit{Addrs} $ ${\Ra \exists q> 0.\ (\gamma(x),q)\in
\pts(x))}$.
\end{enumerate}
\end{definition}

A way to calculate an upper bound for a set of $n$ points-to types
is introduced in the following definition.

\begin{definition}\label{nabla}
Suppose $\pts_1,\ldots,\pts_n$ is a sequence of $n$ points-to types
and $0\le q_1,\ldots,q_n\le 1$ is a sequence of $n$ numbers  whose
sum is less than or equal to $1$. Then
$\nabla((\pts_1,q_1),\ldots,(\pts_n,q_n)):\textit{Var} \ra
2^{\textit{Addrs}_p}$ is the function defined by:
\[\nabla((\pts_1,q_1),\ldots,(\pts_n,q_n))(x)=\]\[\{(z^\prime,p)
\mid (\exists i.\ z^\prime\in A_{\pts_i}(x)) \wedge
(p=\Sigma_{(z^\prime,p_k)\in \pts_k(x)} q_k\times p_k)\}.\]
\end{definition}

We note that the order of the points-to lattice is the point-wise
inclusion. However probabilities are implicitly taken into account
in the definition of supremum which is based on
Definition~\ref{nabla}. Letting the probabilities of points-to
relations be involved in the definition of the order relation
complicates the formula of calculating the lattice supremum. Besides
that this complication is not desirable, introducing probabilities
apparently does not improve the type system results. The definition
for $(\gamma,p)\models\pts$ makes sure that a variable that has an
address under $\gamma$ is allowed (positive probability) to contain
the same address under $\pts$. As for Definition~\ref{nabla}, we can
interpret the elements of the sequence $q_1,\ldots,q_n$ as weights
for the elements of the sequence $\pts_1,\ldots,\pts_n$,
respectively. Therefore the map
$\nabla((\pts_1,q_1),\ldots,(\pts_n,q_n))$ joins
$\pts_1,\ldots,\pts_n$ into one type with respect to the weights.

The following lemma proves that the upper bound of the previous
definition is indeed a points-to type.
\begin{lemma}\label{lem3}
The map $\nabla((\pts_1,q_1),\ldots,(\pts_n,q_n))$ of previous
definition is a points-to type.
\end{lemma}
\begin{proof}
Suppose that $\nabla((\pts_1,q_1),\ldots,(\pts_n,q_n))(x)=
\{(z_1^\prime,t_1),(z_2^\prime,t_2),\ldots,(z_m^\prime,t_m)\}$. To
show the required we need to show that (a) $ 0\le t_i\le 1$ and (b)
$0\le \Sigma_i t_i\le 1$. Since (b) implies (a), it is enough to
show (b). Suppose that $\forall 1\le i\le n,
\pts_i(x)=\{(z_1^\prime,p_{1i}),(z_2^\prime,p_{2i}),\ldots,(z_m^\prime,p_{mi})\}$,
where $\forall 1\le j\le m,\ p_{ji}=0 $ if $z_j\notin A_{pts_i}(x)$.
Then according to Definition~\ref{nabla} the values $t_1,\ldots,t_m$
can be equivalently calculated by the matrix multiplication of
Figure~\ref{mat}.
\begin{figure*}
    \[
\begin{tabular}{l}
$\begin{array}{cccc}
    \qquad \qquad \pts_1 &\quad
    \pts_2 &\quad\ \ \ \ \  \ldots &\quad \ \pts_n\\
  \end{array}$\\
      $\qquad\begin{array}{c}
    z^\prime_1 \\
    z^\prime_2 \\
    \vdots \\
    z^\prime_m \\
  \end{array}
\left[
  \begin{array}{c c c c}
    p_{11}& p_{12} & \ldots & p_{1n} \\
    \  p_{21}  \ &\ \ \ \ \ p_{22}\ \ \ \ \
    &\ \ \ \ \ \ldots\ \ \ \ \  & p_{2n} \\
    \vdots & \vdots & \ddots & \vdots \\
    p_{m1} & p_{m2} & \ldots & p_{mn} \\
  \end{array}
\right]\left(
  \begin{array}{c}
    q_1 \\
    q_2 \\
    \vdots \\
    q_n \\
  \end{array}
\right)=\left(
  \begin{array}{c}
    t_1 \\
    t_2 \\
    \vdots \\
    t_m \\
  \end{array}
\right)$
   \end{tabular}
\]\caption{A matrix multiplication needed in the proof of Lemma~\ref{lem3}.}
\label{mat}
\end{figure*}
Then \begin{eqnarray*}
  \Sigma_i\ t_i &=& (\Sigma_i\ q_i\times p_{1i})+(\Sigma_i\ q_i\times
p_{2i})+\ldots +\\ & & (\Sigma_i\ q_i\times p_{in})  \\
   &=& (q_1\times \Sigma_i\ p_{i1})+(q_2\times \Sigma_i\ p_{i2})
   +\ldots+\\ & & (q_n\times \Sigma_i\ p_{in}).
\end{eqnarray*}
We note that $\forall j,\ 0\le \Sigma_i\ p_{ij}\le 1$ by definition
of $\pts_j$ and $\forall j,\ 0\le q_j\le 1$. Therefore this last
summation is less than $1$.
\end{proof}

\begin{lemma}
Suppose that $A=\{\pts_1,\ldots,\pts_n\}\subseteq \textit{PTS}$ and
$\textit{pts}=\nabla((\pts_1,\frac{1}{n}),\ldots,(\pts_n,\frac{1}{n}))$.
Then with respect to definitions of $\nabla$, the subtyping, and
equality relations introduced in
Definitions~\ref{order}.2,~\ref{order}.3, and~\ref{nabla},
respectively, the set \textit{PTS} is a complete lattice where $\vee
A=pts$.
\end{lemma}
\begin{proof}
Clearly $\pts$ is an upper bound for $A$. Moreover for every $x$,
$A_\pts(x)=\cup_i A_{\pts_i}(x)$. Therefore $\pts$ is the least
upper bound of $A$.
\end{proof}

The inference rules of our proposed type system for probabilistic
pointer analysis are shown in Figure~\ref{protype}.

\begin{figure*}
{\footnotesize {
\[
\begin{prooftree} \justifies n:
\pts\ra \emptyset \thickness=0.08em
\end{prooftree}
\quad
\begin{prooftree}
\justifies x: \pts\ra \pts(x) \thickness=0.08em
\end{prooftree}
\quad
\begin{prooftree}
\justifies e_1\oplus e_2: \pts\ra\emptyset \thickness=0.08em
\end{prooftree}
\quad
\begin{prooftree}
e:\pts\ra A\justifies x \coloneqq e: \pts\ra \pts[x\mapsto
A]\thickness=0.08em\using{(\coloneqq^{prob})}
\end{prooftree}
\]\[
\begin{prooftree}
\pts(y)=\{(z_1^\prime,p_1),\ldots,(z_n^\prime,p_n)\}\quad \forall
i.\ x\coloneqq z_i: \pts\ra\pts_i\justifies x \coloneqq *y: \pts\ra
\pts[x\mapsto
\nabla((\pts_1,p_1),\ldots,(\pts_n,p_n))(x)]\thickness=0.08em
\using{(\coloneqq *^{prob})}
\end{prooftree}
\quad
\begin{prooftree}
\justifies \textit{skip}: \pts\ra \pts\thickness=0.08em
\end{prooftree}
\]\[
\begin{prooftree}
\pts(x)=\{(z_1^\prime,p_1),\ldots,(z_n^\prime,p_n)\}\quad\forall
z_i^\prime \in A_\pts(x).\ z_i\coloneqq e: \pts\ra\pts_i \justifies
*x \coloneqq e: \pts\ra
\pts[z_i\mapsto\nabla((\pts,1-p_i),(\pts_i,p_i))(z_i) \mid
z_i^\prime \in
A_\pts(x)]\thickness=0.08em\using{(*\coloneqq^{prob})}
\end{prooftree}
\]\[
\begin{prooftree}
\justifies x \coloneqq \& y: \pts\ra \pts[x\mapsto
\{(y^\prime,1)\}]\thickness=0.08em \using{(\coloneqq \&^{prob})}
\end{prooftree}
\quad
\begin{prooftree} S_1: \pts \ra \pts^{\prime\prime}
\quad S_2: \pts^{\prime\prime} \ra \ptsp\justifies S_1;S_2: \pts\ra
\ptsp \thickness=0.08em\using{(\textit{seq}^{prob})}
\end{prooftree}
\]\[
\begin{prooftree}
S_t:\pts\ra \pts_t \quad S_f:\pts\ra \pts_f\justifies \ifs: \pts\ra
\nabla((\pts_t,p),(\pts_f,1-p))\thickness=0.08em\using{(\textit{if}^{prob})}
\end{prooftree}
\]\[
\begin{prooftree}
S_i: \nabla\{(\pts,1/n),(\pts_j,1/n)\mid
j\not=i\}\ra\pts_i\justifies \textit{par}\{\{S_1\},\ldots,\{S_n\}\}:
\pts\ra \nabla((\pts_1,1/n),\ldots,(\pts_n,1/n)) \thickness=0.08em
\using{(\textit{par}^{prob})}
\end{prooftree}
\]\[
\begin{prooftree}
\textit{par}\{\{\mathit{if\ b_1\ then\ }S_1 \mathit{\ else\
skip}\},\ldots,\{\mathit{if\ b_n\ then\ }S_n \mathit{\ else\
skip}\}\}: \pts\ra \ptsp\justifies
\textit{par-if}\{(b_1,S_1),\ldots,(b_n,S_n)\}: \pts\ra
\ptsp\thickness=0.08em \using{(\textit{par-if}^{prob})}
\end{prooftree}
\]\[
\begin{prooftree}
\forall n.\
\textit{par}\{\overbrace{\{S\},\ldots,\{S\}}^{n-times}\}:
\pts\ra\ptsp\justifies \textit{par-for}\{S\}: \pts\ra
\ptsp\thickness=0.08em
\using{(\textit{par-for}^{prob})}\end{prooftree} \quad
\begin{prooftree}
n=0 \justifies \while: \pts\ra \pts
\thickness=0.08em\using{(\textit{whl}^{prob}_1)}
\end{prooftree}\]\[
\begin{prooftree}
n\ge 1\qquad\forall 1\le i\le n.\ S_t: \pts\ra^i \pts_i \justifies
\while: \pts\ra \nabla((\pts_1,1/n),\ldots,(\pts_n,1/n))
\thickness=0.08em\using{(\textit{whl}^{prob}_2)}
\end{prooftree}
\]\[
\begin{prooftree}
\ptsp_1\leq \pts_1\quad S:\pts_1 \ra \pts_2 \quad \pts_2 \leq
\ptsp_2\justifies S: \ptsp_1\ra
\ptsp_2\thickness=0.08em\using{(\textit{csq}^{prob})}
\end{prooftree}
\]
}} \caption{The inference rules for the type system for
probabilistic pointer analysis}\label{protype}
\end{figure*}
The judgment of an arithmetic expression has the form $e:\pts \ra
A$. The intuition (Lemma~\ref{lem2}) of this judgment is that any
address that $e$ evaluates to in a state of type $pts$ is included
in the set $A$ as the second component of a pair whose first
component is a non-zero probability. The judgment for a statement
$S$ has the form ${S:\pts\ra\ptsp}$ and guarantees that if the
execution of $S$ in a state of type $\pts$ terminates then the
reached state is of type $\ptsp$. This is proved in
Theorem~\ref{soundness-points to 2}.

Concerning the inference rules, some comments are in order. In the
rule $(\coloneqq *^{prob})$, since there are $n$ possible ways to
modify $x$, the post-type is calculated from the pre-type by
assigning $x$ its value according to the upper bound of the $n$
ways. The upper bound is consider to enable the analysis to cover
all possible executions of the statement. In the rule
$(*\coloneqq^{prob})$, there are $n$ variables,
$\{z_1,\ldots,z_n\}$, that have a chance of getting modified. This
produces $n$ post-types in the pre conditions of the rule. Therefore
the post-type is calculated from the pre-type by assigning each of
the $n$ variables its image under the upper bound of the $n$
post-types. In the rule $(\textit{if}^{prob}), p$ is the probability
that the condition of the \textit{if} statement is true. The rule
$(\textit{par}^{prob})$ has this form in order for the analysis
result of any thread $S_i$ of the \textit{par} statement to consider
the fact that any other thread may have been executed before the
thread in hand. As it is the case in the operational semantics, the
rules for conditionally spawned threads $(\textit{par-if}^{prob})$
and parallel loops $(\textit{par-for}^{prob})$ are built on the rule
$(\textit{par}^{prob})$. In the following we give an example for the
application of the rule $(\textit{par}^{prob})$. Let:
\begin{itemize}
    \item $S_1: \textit{if}\ b_1\textit{ then } x\coloneqq \&y
    \textit{ else } x\coloneqq 5,$
    \item $S_2:x\coloneqq \&z;$
    \item $S_{\textit{par}}:\textit{par}\{\{S_1\},\{S_2\}\},$
    \item $\pts=\{t\mapsto\emptyset\mid t\in\textit{Var}\},$ $
    {\pts_1=\{x\mapsto\{(y^\prime,0.4)\},t\mapsto\emptyset
    \mid x\not=t\in\textit{Var}\}},$
    and $\pts_2=\{x\mapsto\{(z^\prime,1)\},t\mapsto\emptyset\mid
x\not=t\in\textit{Var}\}$.
\end{itemize} We suppose that the condition $b_1$ in $S_1$
succeeds with probability $0.4$. Then we have the following:
\begin{itemize}
    \item $\nabla((\pts,1/2),(\pts_1,1/2))={\{x\mapsto\{(y^\prime,0.25)\},
    t\mapsto\emptyset\mid x\not=t\in\textit{Var}\}}$,
    \item $\nabla((\pts,1/2),(\pts_2,1/2))={\{x\mapsto\{(z^\prime,0.5)\},
    t\mapsto\emptyset\mid x\not=t\in\textit{Var}\}}$, and
    \item $\nabla((\pts_1,1/2),(\pts_2,1/2))={\{x\mapsto\{(y^\prime,0.25),
    (z^\prime,0.5)\}, t\mapsto\emptyset\mid x\not=t\in\textit{Var}\}}$.
\end{itemize}
Clearly, ${S_1:\nabla((\pts,1/2),(\pts_2,1/2))\ra \pts_1}$ and
$S_2:\nabla((\pts,1/2),(\pts_1,1/2))\ra \pts_2$. These two judgments
constitute the hypotheses for the rule $(\textit{par}^{prob})$.
Therefore using the rule $(\textit{par}^{prob})$, we can conclude
that ${S_{\textit{par}}:\pts\ra\nabla((\pts_1,1/2),(\pts_2,1/2))}$.
The post type of $S_{\textit{par}}$ clearly covers all semantics
states that can be reached by executing $S_{\textit{par}}$. Now we
give an example for the application of the rule
$(\textit{par-if}^{prob})$. Let:
\begin{itemize}
    \item $S_1: x\coloneqq \&y,$
    \item $S_2:x\coloneqq \&z,$
    \item $S_{\textit{par-if}}:\textit{par-if}\{(b_1,S_1),(\textit{true},S_2)\},$ and
    \item $\pts^\prime=\{x\mapsto\{(y^\prime,0.25), (z^\prime,0.5)\}
    ,t\mapsto\emptyset\mid {x\not=t\in\textit{Var}}\}$ and
    $\pts=\{t\mapsto\emptyset\mid t\in\textit{Var}\}$.
\end{itemize} We suppose that the condition $b_1$
succeeds with probability $0.4$. By the previous example it should
be clear that $\textit{par}\{\{\textit{if}\ b_1\ \textit{then } S_1
\textit{ else skip}\},\{\textit{if } \textit{true}\ \textit{then }
S_2 \textit{ else skip}\}\}:\pts\ra\ptsp$. This last judgment
constitutes the hypothesis for the rule $(\textit{par-if}^{prob})$.
Therefore using the rule $(\textit{par-if}^{prob})$, we can conclude
that $S_{\textit{par-if}}:\pts\ra\ptsp$. The post type of
$S_{\textit{par-if}}$ clearly covers all semantics states that can
be reached by executing $S_{\textit{par-if}}$. In rules
$(\textit{whl}^{prob}_1)$ and $(\textit{whl}^{prob}_2), n$
represents an upper bound for the trip-count of the loop. The
post-type of $(\textit{whl}^{prob}_2)$ is an upper bound for
post-types resulting for all number of iterations bounded by $n$.

The proof of the following lemma is straightforward.
\begin{lemma}\label{lem2}
\begin{enumerate}
    \item $\pts\le \ptsp\Ra (\forall (\gamma,p). $ ${(\gamma,p)\models \pts \Ra
    (\gamma,p)\models \ptsp)}$
    \item Suppose $e:\pts \ra A$ and $(\gamma,p)\models \pts$.
    Then $\lb e\rb\gamma \in \textit{Addrs}$ implies
    $(\lb e\rb\gamma,q) \in A$, for some $q>0$.
\end{enumerate}
\end{lemma}

Lemma~\ref{lem2}.1 formalizes the soundness of points-to types.
Lemma~\ref{lem2}.2 shows that for a certain state that is of a
certain type, if the evaluation of an expression with respect to the
state is an address, then this evaluation is surely (positive
probability) approximated by the evaluation of the expression with
respect to the type.

The following theorem proves the soundness of the type system. The
meant soundness implies that the type system respects the
operational semantics with respect to the relation $\models$ whose
definition is based on probabilities.
\begin{theorem} \( (Soundness)\)\label{soundness-points to 2}
Suppose that ${S:\pts\ra \ptsp},$ $ {S:(\gamma,p)\rightsquigarrow
(\gamma^\prime,p^\prime)}$, and ${(\gamma,p)}\models {pts}$. Then
${(\gamma^\prime,p^\prime)} \models {\ptsp}$.
\end{theorem}
\begin{proof}
A structure induction on type derivation can be used to complete the
proof of this theorem. Some cases are presented below.
\begin{itemize}
\item The case of $(\coloneqq^{prob})$: in this case
$p^\prime=p,\ \ptsp=pts[x\mapsto A]$, and
$\gamma^\prime=\gamma[x\mapsto \lb e\rb\gamma]$. Hence by
Lemma~\ref{lem2}.2, $\gamma \models (\pts,p)$ implies $\gamma^\prime
\models (\ptsp,p^\prime)$.
\item The case of $(\coloneqq*^{prob})$: in this case for some
$z\in \textit{Var},\ \gamma(y)=z^\prime$ and $x \coloneqq z:
(\gamma,p)\rightsquigarrow (\gamma^\prime,p)$. For some $i,\
z^\prime=z_i^\prime$ since $(\gamma,p) \models \pts$. Hence by
assumption $x\coloneqq z_i: \pts\ra\pts_i$. Therefore by soundness
of $(\coloneqq^{prob})$, $(\gamma^\prime,p)\models \pts_i\le
\ptsp=\pts[x\mapsto \nabla((\pts_1,p_1),\ldots,(\pts_n,p_n))(x)]$.
\item The case of $(*\coloneqq^{prob})$: in this case there exists
$z\in \textit{Var}$ such that $\gamma(x)=z^\prime$ and $z\coloneqq
e:(\gamma,p)\rightsquigarrow (\gamma^\prime,p)$. For some $i,\
z^\prime=z_i^\prime$ since $(\gamma,p) \models \pts$. Hence by
assumption $z_i\coloneqq e: \pts\ra\pts_i$. Therefore by soundness
of $(\coloneqq^{prob})$, $(\gamma^\prime,p)\models \pts_i\le
\ptsp=\pts[z_i\mapsto\nabla((\pts,1-p_i),(\pts_i,p_i))(z_i) \mid
z_i^\prime \in A_\pts(x)]$.
\item The case of $(par^{prob})$: in this case there exist a permutation
$\theta:\{1,\ldots,n\}\ra\{1,\ldots,n\}$ and $n+1$ states
$(\gamma_1,p_1),\ldots,(\gamma_{n+1},p_{n+1})$ such that
$(\gamma,p)=(\gamma_1,p_1),
(\gamma^\prime,p^\prime)={(\gamma_{n+1},\frac{1}{n!}\times
p^\prime_{n+1})}$, and for every ${1\le i\le n},\
{S_{\theta(i)}:(\gamma_i,p_i)\ra (\gamma_{i+1},p_{i+1})}$. Also
${(\gamma_1,p_1)\models \pts \le \nabla\{(\pts,1/n),(\pts_j,1/n)\mid
j\not=1\}}$. Therefore by the induction hypothesis
$(\gamma_2,p_2)\models \pts_1\le
{\nabla\{(\pts,1/n),(\pts_j,1/n)\mid j\not=2\}}$. Again by the
induction hypothesis we get $(\gamma_3,p_3)\models \pts_2$.
Therefore by a simple induction on $n$, we can show that
$(\gamma_{n+1},p_{n+1})\models \pts_n\le
\nabla((\pts_1,1/n),\ldots,(\pts_n,1/n))=\ptsp$. This implies
$(\gamma^\prime,p^\prime)=(\gamma_{n+1},\frac{1}{n!}\times
p^\prime_{n+1})\models\ptsp$
\item The case of $(par-for^{prob})$: in this case there exists $n$ such that $
{\textit{par}\{\overbrace{\{S\},\ldots,\{S\}}^{n-times}\}:
(\gamma,p)\rightsquigarrow(\gamma^\prime,p^\prime)}$. By induction
hypothesis we have ${
\textit{par}\{\overbrace{\{S\},\ldots,\{S\}}^{n-times}\}:\pts\ra
\ptsp}$. Therefore by the soundness of $(par^{prob})$,
$(\gamma^\prime,p^\prime)\models \ptsp$.
\item The case of $(\textit{whl}^{prob}_2)$: in this case there
exist $m\le n$ and $m+1$ states, $(\gamma_1,p_1),\ldots,
(\gamma_{m+1},p_{m+1})$, such that $(\gamma,p)=(\gamma_1,p_1),\
(\gamma^\prime,p^\prime)=(\gamma_{m+1},p_{m+1}),\ $ and $\forall
1\le i\le m.\ S:
(\gamma_i,p_i)\rightsquigarrow(\gamma_{i+1},p_{i+1})$. By induction
hypothesis we have $(\gamma^\prime,p^\prime)\models \pts_m \le
\nabla((\pts_1,1/n),\ldots,(\pts_n,1/n))$. Therefore
$(\gamma^\prime,p^\prime)\models \ptsp$ as required.
\end{itemize}
\end{proof}

We note that probabilities are mentioned implicitly in
Theorem~\ref{soundness-points to 2}. This is in the condition that
${(\gamma,p)}\models {pts}$. Some of the implications of this
implicit consideration of probabilities are explicit in
Lemma~\ref{lem2}.2. As an example for the theorem, executing the
statement $S_{\textit{par}}$, defined above, from the semantics
state $\gamma=\{t\mapsto 0\mid t\in \textit{Var}\}$ may result in
the state $\gamma^\prime=\{t\mapsto 0,x\mapsto z^\prime\mid
x\not=t\in \textit{Var}\}$. This happens if $S_2$ is executed after
$S_1$. Clearly we have that $\gamma\models \{t\mapsto\emptyset\mid
t\in\textit{Var}\}$ and $\gamma^\prime\models
\{x\mapsto\{(y^\prime,0.25),
    (z^\prime,0.5)\}, t\mapsto\emptyset\mid x\not=t\in\textit{Var}\}$.

One source of attraction in the use of type systems for program
analysis is the relative simplicity of the inference rules. This
simplicity is very important when practical implementation is
concerned. The simplicity of the rules naturally simplifies
implementations of rules and hence the type system. In particular,
from experience related to coding similar type systems, we believe
that the implementation of the type system presented in this paper
is straightforward and efficient in terms of space and time.
\section{Related work}\label{rwork}
\subsubsection{Analysis of multithreaded programs:}
Typically, analyses of multithreaded programs are classified into
two main categories: (a) techniques that were originally designed
for sequential programs and later extended to analyze multithreaded
programs and (b) techniques that were designed specifically for
analyzing, optimizing, or correcting multithreaded programs.

The first category includes flow-insensitive approaches providing an
easy way to analysis multithreaded programs. This is done via
considering all possible combinations of statements used in a
parallel structure. The drawback of this approach is that it is not
practical enough due to huge number of combinations. However
flow-sensitive approaches of sequential programs were also extended
to cover multithreaded programs. Examples of these techniques are
constant propagation~\cite{Lee97}, code motion~\cite{Knoop99}, and
reaching definitions~\cite{Sarkar09}.

The category of techniques that were designed specifically for
multithreaded programs include deadlock detection, data race
detection, and weak memory consistency. A round abeyance to gain
resources usually results in a deadlock
situation~\cite{Kim09,Wang08,Xiao10}. Synchronization analysis is a
typical start to study deadlock detection for multithreaded
programs. In absence of synchronization, if two parallel threads
write to the same memory location, a situation of a data
race~\cite{Leung09} results. Data race analyses aim at eliminating
data race situations as they are mainly programmer error. Models of
weak memory consistency~\cite{Gelado10} aims at improving
performance of hardware. This improvement usually results in
complicating parallel programs construction and analysis.
\subsubsection{Probabilistic pointer analysis and speculative optimizations:}
Although pointer analysis is a well-established program analysis and
many techniques have been suggested, there is no single technique
that is believed to be the best choice~\cite{Hind00}. The trade-off
between accuracy and time-costs hinders a universal pointer analysis
and motivates application-directed techniques for pointer
analysis~\cite{Hind01}. A probabilistic pointer analysis that is
flow-sensitive and context-insensitive is presented in~\cite{Sun11}
for Java programs. While our work is based on type systems, the work
in~\cite{Sun11} is based on interprocedural control flow graphs
(ICFG) whose edges are enriched with probabilities. While our work
treats multhithreaded programs, the work in~\cite{Sun11} treats only
sequential programs. Context-sensitive and control-flow-sensitive
pointer analyses~\cite{Hardekopf09,Wang09,Yu10,El-Zawawy11} are
known to be accurate but not scalable. On the other hand the
context-insensitive control-flow-insensitive
techniques~\cite{Adams02,Anderson02} are scalable but excessively
conservative. A convenient mixture of accuracy and scalability is
introduced by some technique~\cite{Berndl03,Whaley04,Zhu04} to
optimize the trade-off mentioned above. The probabilistic pointer
analysis of a simple imperative language and the pointer analysis of
multithreaded programs were studied in~\cite{Chen04,Silva06}
and~\cite{Rugina03}, respectively. However none of these typical
techniques for pointer analysis study the probabilistic pointer
analysis of multithreaded programs.

Speculative
optimizations~\cite{Ramalingam96,Ju99,Fernandez02,Bhowmik03} are
considered by many program analyses. A probabilistic technique for
memory disambiguation was proposed in~\cite{Ju99}. This technique
measures the probability that two array references alias.
Nevertheless this approach is not convenient to pointers. By
lessening the safety of analysis, the work in~\cite{Fernandez02}
introduces a pointer analysis that considers speculation. Another
unsafe analysis, which achieves scalability using transfer
functions, is proposed in~\cite{Bhowmik03}. The problem with these
last two approaches is that they do not compute the probability
information required by speculative optimizations.

\subsubsection{Type systems in program analysis:}
There are general
algorithms~\cite{Laud06,Benton04,Nielson02,Saabas08,El-Zawawy11,Nicola10}
for using type systems to present dataflow analyses, which are
monotone and forward or backward. While a
way~\cite{Saabas08,Benton04} to reason about program pairs using
relational Hoare logic exists, program
optimizations~\cite{Saabas08,Nielson02} as types systems also exist.
Type systems were also used to cast safety policies for resource
usage, information flow, and carrying-code
abstraction~\cite{Beringer04,Besson06}. Proving the soundness of
compiler optimizations for imperative languages, using type systems,
gained much interest~\cite{Bertot04,Laud06,Saabas08} of many
researchers. Other work studies translating proofs of functional
correctness using wp-calculus~\cite{Barthe09} and using a Hoare
logic~\cite{Saabas08}. There are other
optimizations~\cite{Aspinall07} that boost program quality besides
maintaining program semantics.
\subsubsection{Edge and path profiling:}
Edge (path) profiling research simply aims at profiling programs
edges (paths). The profiling process can be done statically or
dynamically. Profiling techniques can be classified into:
\begin{itemize}
    \item Sample-based techniques~\cite{Anderson97,Suganuma05} which
profile representative parts of active edges and paths,
    \item One-time profiling methods which profile only part of the execution of the
program to cut down the overhead~\cite{Suganuma05,Zilles02},
    \item Instrumentation-based techniques~\cite{Joshi04} which are more
convenient for programs with comparably anticipated behavior, and
    \item Hardware profiling which employs hardware to gather edge profiles
    using existing hardware for branch anticipation~\cite{Vaswani05}.
\end{itemize}
Using a parallel data-flow diagram~\cite{Grunwald93}, many of these
techniques are applicable to the language studied in this paper. In
particular the technique presented in~\cite{Bond05}, a hybrid
sampling and instrumentation approach, is a convenient choice giving
its simplicity and powerful.

\section{Acknowledgments} This work was started during the author's
sabbatical at Institute of Cybernetics, Estonia in the year 2009.
The author is grateful to T. Uustalu for fruitful discussions. This
work was partially supported by the EU FP6 IST project MOBIUS. The
author is also indebted to the anonymous reviewers whose queries and
comments improved the paper.

\end{document} 


\begin{equation}

\label{e:}
\end{equation}

\begin{table} 
\caption{}
\label{t:}
\begin{tabular}{lrrrr}
\hline
 &  &  &  &   \\
\hline
\end{tabular}\\
\end{table}

\begin{table*} 
\caption{}
\label{t:}
\begin{tabular}{lrrrr}
\hline
 &  &  &  &   \\
\hline
\end{tabular}\\
\end{table*}

\begin{figure} 
\centering{\includegraphics[width=5cm]{}}
\caption{}
\label{f:} 
\end{figure}

\begin{figure*} 
\centering{\includegraphics[width=12cm]{}}
\caption{}
\label{f:} 
\end{figure*}